\documentclass[11pt]{article}
\usepackage[english]{babel}
\usepackage{amsmath,amsthm}
\usepackage{amsfonts}
\usepackage[comma,numbers,square,sort&compress]{natbib}
\usepackage{graphicx,subfigure,amsmath,amssymb,mathrsfs,amsfonts,amstext,amsthm}
\usepackage{latexsym, amssymb}
\usepackage{graphicx}
\usepackage{subfigure}
\usepackage{pgf,tikz}
\usepackage{pstricks}
\newtheorem{thm}{Theorem}[section]
\newtheorem{cor}[thm]{Corollary}
\newtheorem{lem}[thm]{Lemma}

\theoremstyle{definition}
\newtheorem{defn}[thm]{Definition}
\theoremstyle{remark}
\newtheorem{rem}[thm]{Remark}
\numberwithin{equation}{section}
\begin{document}

\title{\bfseries\textrm{Quasi-synchronization of bounded confidence opinion dynamics with stochastic asynchronous rule*}
\footnotetext{*This research was supported by the National Key Research and Development Program of Ministry of Science and Technology of China under Grant No. 2018AAA0101002, the National Natural Science Foundation of China under grants No. 61803024, 11688101,61906016,  the General Project of Scientific Research Project of the Beijing Education Committee under Grant No. KM201811417002, the Foundation of Beijing Union University under Grant No. BPHR2019DZ08,
the Young Elite Scientists Sponsorship Program by CAST under Grant 2018QNRC001, and Beijing Institute of Technology Research Fund Program for Young Scholars.\\
\,\,Wei Su is with School of Automation and Electrical Engineering, University of Science and Technology Beijing \& Key Laboratory of Knowledge Automation for Industrial Processes, Ministry of Education, Beijing 100083, China, {\tt suwei@amss.ac.cn}. Xueqiao Wang is with Beijing Key Laboratory of Information Service Engineering, Beijing Union University, Beijing 100101, China, {\tt ldxueqiao@buu.edu.cn}. Ge Chen is with National Center for Mathematics and Interdisciplinary Sciences \& Key Laboratory of Systems and Control, Academy of Mathematics and Systems Science, Chinese Academy of Sciences, Beijing 100190, China, {\tt chenge@amss.ac.cn}. Kai Shen is with School of Automation, Beijing Institute of Technology, Beijing, 100081, China, {\tt kshen@bit.edu.cn}.}}
\author{Wei Su \and Xueqiao Wang \and Ge Chen \and Kai Shen}

%
\date{}%
\maketitle
\begin{abstract}
Recently the theory of noise-induced synchronization of Hegselmann-Krause (HK) dynamics has been well developed. As a typical opinion dynamics of bounded confidence, the HK model obeys a synchronous updating rule, i.e., \emph{all} agents check and update their opinions at each time point. However, whether asynchronous bounded confidence models, including the famous Deffuant-Weisbuch (DW) model, can be synchronized by noise have not been theoretically proved. In this paper, we propose a generalized bounded confidence model which possesses a stochastic asynchronous rule. The model takes the DW model and the HK model as special cases and can significantly generalize the bounded confidence models to practical application. We discover that the asynchronous model possesses a different noise-based synchronization behavior compared to the synchronous HK model. Generally, the HK dynamics can achieve quasi-synchronization \emph{almost surely} under the drive of noise. For the asynchronous dynamics, we prove that the model can achieve quasi-synchronization \emph{in mean}, which is a new type of quasi-synchronization weaker than the ``almost surely'' sense. The results unify the theory of noise-induced synchronization of bounded confidence opinion dynamics and hence proves the noise-induced synchronization of DW model theoretically for the first time. Moreover, the results provide a theoretical foundation for developing noise-based control strategy of more complex social opinion systems with stochastic asynchronous rules.
\end{abstract}

\textbf{Keywords}:Quasi-synchronization in mean, noise, asynchronous, bounded confidence, stochastic opinion dynamics
\section{Introduction}\label{intro}

In recent years, opinion dynamics are attracting increasing attention of researchers in various areas \cite{Proskurnikov2017,Proskurnikov2018,Ye2018,Chen2019}. One of the interesting topics in opinion dynamics is the noise-induced synchronization of opinion systems. During the study of opinion dynamics in a noisy environment, it has been found in a mass of simulations that the bounded confidence opinion models display a very positive tendency to be synchronized under the disturbance of noise \cite{Mas2010,Carro2013,Pineda2013,HadStauffHan2015}. Using the widely known Hegselmann-Krause (HK) model \cite{Hegselmann2002}, Su et al made a rigorous theoretical analysis of this phenomenon \cite{Su2017auto}. They proved that the noisy HK model can achieve a type of quasi-synchronization (a definition of synchronization with noise) in finite time almost surely for any initial state. After that a series of conclusions and applications of noise-based synchronization of opinion systems were obtained, including the noise-induced truth seeking \cite{Su2017free}, and the noise-induced synchronization of HK model in full space \cite{Su2019tac} and also the heterogeneous HK model \cite{Chen2019tac}, etc. These simulation and theoretical studies effectively reveal the opinion evolution in noisy circumstances. More crucially, due to the advantages of random information working as an effective social opinion control method \cite{Su2019noisecontrol}, the well establishment of noise-induced synchronization of opinion dynamics could provide a theoretical foundation for developing noise-based control strategies of social opinion system.

The analysis of noise-induced synchronization mentioned above are mainly based on the HK model, which is a typical bounded confidence opinion model. Notably, the HK model obeys a synchronous opinion updating rule, i.e., all agents continuously check and decide whether or not they will update their opinions at each time. However, a more practical case in reality is that people are more likely to communicate with each other in an asynchronous way, where only a random fraction of agents are communicating with each other at each time. Hence a more general asynchronous model is needed. An extreme example of asynchronous opinion dynamics is the famous Deffuant-Weisbuch (DW) model \cite{Deffuant2000}, where only two agents are randomly selected to communicate and therefore an asynchronous updating rule is adopted. For asynchronous models of bounded confidence, only a fraction of agents (we may call them as communicating agents) update their opinions according to the bounded confidence mechanism at each time point. Obviously, the asynchronous rule presents a more practical process of opinion evolution, especially for a large-scale social system. Currently, there are several literatures that studied the opinion formation with asynchronous rule. Touri et al \cite{Touri2014} and Etesami et al \cite{Etesami2015} both considered an asynchronous HK model, in which only one agent is communicating at each time. Rossi et al studied an opinion dynamics on the
$k$-nearest-neighbors graph, which is essentially an asynchronous model \cite{Rossi2018}. Ding et al considered an asynchronous opinion scenario with online and offline interactions \cite{Ding2017}.

At the same time, there are some theoretical studies that focus on the noise-driven properties of asynchronous opinion models. Just as the HK model, simulation studies revealed that noise-driven synchronization also arises in the DW model \cite{Carro2013}. However, using the method of analyzing noisy HK model, we can show that the DW model cannot achieve quasi-synchronization almost surely as the HK model does. Baccelli et al studied a noisy DW model on graph and proved some sufficient conditions for the system to converge. The sufficient conditions given by the authors are essentially related to the state connectivity of the system, which is difficult to be predetermined in bounded confidence models. Zhang et al carried out a analysis of consensus property for a noisy DW model in \cite{Zhang2018}, where the definition of $T$-robust consensus was introduced. The authors proved that the noisy DW model can achieve $T$-robust consensus a.s. The results imply that the noisy DW model can achieve a synchronization state at some moment $T^*\geq 0$ a.s. and keep synchronized in a time interval $[T^*,T^*+T]$. In addition, the authors showed that there is a positive probability with which robust consensus happens. With these pioneering studies, however, it is still unclear whether the DW model as well as the more general asynchronous models of bounded confidence possess a property of noise-induced synchronization, and a  complete theoretical conclusion is far to be obtained.

In this paper, we intend to make a theoretical analysis of noise-induced synchronization in depth for general asynchronous opinion dynamics. Firstly, we propose a generalized stochastic asynchronous model, for which the HK model and DW model may be considered as its special examples. In the stochastic asynchronous model, communicating agents are randomly chosen at each time, and then it is possible to update their opinions according to the bounded confidence mechanism. Interestingly, we prove that the stochastic asynchronous model cannot achieve the noise-driven synchronization almost surely as the HK model does. Meanwhile, simulation studies suggest that synchronization in some sense emerges in the presence of noise. Finally, by introducing the definition of \emph{quasi-synchronization in mean}, we prove that the stochastic asynchronous model will achieve quasi-synchronization in mean under the drive of noise. This result naturally implies the noise-induced synchronization of DW model, i.e., for any initial state, the DW model will achieve quasi-synchronization in mean under the drive of noise.  We also show that the quasi-synchronization in mean is weaker than the quasi-synchronization almost surely. The results hence reveal that the quasi-synchronization in mean we propose here is an essentially new concept which is necessary to study the noise-induced synchronization of asynchronous models. This discovery also accords with the fact that the asynchronous updating rule possesses a lower interaction frequency than that of the synchronous model. At last, a simulation result is demonstrated to understand the quasi-synchronization in mean.
Compared to the existing studies on the noise-induced synchronization of  synchronous HK model, the results in this paper provide essentially new definitions and discoveries, which completes the theory of noise-induced synchronization of bounded confidence opinion dynamics and for the first time solves the noise-induced synchronization problem of DW model. More importantly, the results in this paper lay a theoretical foundation for developing noise-based control strategy of more complex social systems with stochastic asynchronous rules.

The rest of the paper is organized as follows: Section \ref{Section:formulation} presents some preliminaries of models and definition. Section \ref{Section:results} gives the main theoretical results of the paper; Section \ref{Section:simulations} shows a simulation result to illustrate the quasi-synchronization in mean. Finally some concluding remarks are given in Section \ref{Section:conclusions}.

\section{Model description and preliminaries}\label{Section:formulation}

Denote $\mathcal{V}=\{1,2,\ldots,n\}, n\geq 2$ as the set of $n$ agents, $x_i(t)\in[0,1], i\in\mathcal{V}, t\geq 0$ as the opinion state of agent $i$ at time $t$, and let $\mathcal{U}(t)\subset\mathcal{V}$ be the set of communicating agents at time $t$, and $\{\xi_i(t), i\in\mathcal{V}, t\geq 1\}$ be the noise. Usually, noise is used to model the external and internal factors that randomly affect people's opinions, such as the random information flow on social media each day and the intrinsic opinion uncertainty resulting from one's free will. Moreover, let $\Omega$ be the sample space of $\{\xi_i(t),i\in\mathcal{S}, \mathcal{U}(t), t\geq 1\}$, $\mathcal{F}$ be the generated $\sigma-$algebra, and $\mathbb{P}$ be the probability measure on $\mathcal{F}$, so the underlying probability space is written as $(\Omega,\mathcal{F},\mathbb{P})$.
In addition, denote $\textbf{E}\{\cdot\}$ as the expectation of a random variable.

To proceed, we first propose our general bounded confidence model with stochastic asynchronous rule.
Denote
\begin{equation*}
  y_{[0,1]}=\left\{
              \begin{array}{ll}
                0, & \hbox{$y<0$} \\
                y, & \hbox{$0\leq y\leq 1$.} \\
                1, & \hbox{$y>1$}
              \end{array}
            \right.
\end{equation*}
Then, the update rule of the stochastic opinion dynamics yields:
\begin{equation}\label{basicHKmodel}
  x_i(t+1)=\left\{
             \begin{array}{ll}
               \Big(\alpha_i(t)x_i(t)+(1-\alpha_i(t))\frac{\sum_{j\in\mathcal{N}_i(t)}x_j(t)}{|\mathcal{N}_i(t)|}+\xi_i(t+1)\Big)_{[0,1]}, & \hbox{if $i\in\mathcal{U}(t), \mathcal{N}_i(t)\neq \emptyset$;}\\
               (x_i(t)+\xi_i(t+1))_{[0,1]}, & \hbox{otherwise.}
             \end{array}
           \right.
\end{equation}
where
\begin{equation}\label{basicHKmodelalpha}
  \alpha_i(t)\in[\alpha,1-\alpha]
\end{equation}
with $\alpha\in(0,\frac{1}{n}]$ is the inertial coefficient of agent $i$ at $t$,
\begin{equation}\label{neigh}
 \mathcal{N}_i(t)=\{j\in\mathcal{U}(t)-\{i\}\; \big|\; |x_j(t)-x_i(t)|\leq \epsilon\}
\end{equation}
is the neighbor set of $i$ at $t$ and $\epsilon\in(0,1]$ represents the confidence threshold of the agents. Here, $|\cdot|$ is the cardinal number of a set or the absolute value of a real number accordingly.

From (\ref{basicHKmodel}), we can see that when an agent $i\in\mathcal{U}(t)$, i.e., $i$ is a communicating agent at $t$, $i$ will communicate with other communicating agents at $t$, and then updates its opinions according to the bounded confidence mechanism. When $i$ is not a communicating agent at $t$, it will not communicate with other agents, while its opinion can only be affected by environment noise.

Now we discuss the communicating set $\mathcal{U}(t)$. In reality, people seldom update opinion values simultaneously at each time. Hence the elements of $\mathcal{U}(t)$ are randomly selected from $\mathcal{V}$. Formally, we suppose that for $t\geq0$, $\mathcal{U}(t)$ satisfies
\begin{equation}\label{Model:probnum}
\begin{split}
(a)~&\mathbb{P}\{|\mathcal{U}(t)|=k\}=p_k,~\hbox{where}~0\leq k\leq n, 0\leq p_k\leq 1 \,\hbox{and}\,\sum_kp_k=1, 0\leq p_0+p_1<1;\\
(b)~&\hbox{for any}\, i_1,i_2\in\mathcal{V}, i_1\neq i_2,\\
&\mathbb{P}\{i_1,i_2\in\mathcal{U}(t)\}=\mathbb{P}\{i_1\in\mathcal{U}(t)\}\mathbb{P}\{i_2\in\mathcal{U}(t)\},
\mathbb{P}\{i_1\in\mathcal{U}(t)\}=\mathbb{P}\{i_2\in\mathcal{U}(t)\};\\
(c)~&\{\mathcal{U}(t),\xi_i(t),i\in\mathcal{V},t\geq 0\}~\hbox{are mutually independent}.\\
\end{split}
\end{equation}

(\ref{Model:probnum}) provides the assumptions of communicating agents. The assumptions are natural and we will further provide a methodological explanation of them. (\ref{Model:probnum})(a) assumes that the number of communicating agents is randomly selected at each time, and also the number of communicating agents is not always 0 or 1, since we suppose that there do exist communication among agents with a positive probability. (\ref{Model:probnum})(b) assumes that each agent is independently selected to be the communicating agent with an equal probability at each time. With (\ref{Model:probnum})(b) we can have
\begin{equation}\label{Equa:utik}
\begin{split}
  \mathbb{P}\{\mathcal{U}(t)=\{i_1,\ldots,i_k\}\}=\frac{1}{C_n^k}p_k,
\end{split}
\end{equation}
where $\{i_1,\ldots,i_k\}\subset\mathcal{V}$ is an arbitrary choice of $k$ agents and $C_n^k$ is the combinatorial symbol that abbreviates $\frac{n!}{k!(n-k)!}$ for $0\leq k\leq n$. (\ref{Equa:utik}) indicates that given the communicating agents number is $k$, then any choice of a combination of $k$ agents possesses an equal probability.
(\ref{Model:probnum})(c) assumes that the event of which agents are communicating and the noises are mutually independent at each time, and also the updating process of communicating agents and noise are mutually independent for different time point.

The model represented by equations (\ref{basicHKmodel})-(\ref{Model:probnum}) can be considered as a generalized bounded confidence opinion dynamics with a stochastic asynchronous rule, of which the HK model and the DW model are two special examples under some conditions. For example, consider the deterministic case of model (\ref{basicHKmodel})-(\ref{Model:probnum})($\xi_i(t)\equiv0, i\in\mathcal{V}, t\geq 1$). If we let $\mathcal{U}(t)=\mathcal{V}, \alpha_i(t)=\frac{1}{|\mathcal{N}_i(t)|+1}$, the model degenerates into the standard HK model; While, if we let $\mathcal{U}(t)=\{i_1, i_2\}$ where $\{i_1,i_2\}$ is an arbitrary choice of two agents, and $\alpha_i(t)\equiv \beta\in (0,1]$, the model degenerates into the standard DW model.

In order to study the noise-induced synchronization of the model (\ref{basicHKmodel})-(\ref{Model:probnum}), we introduce the definition of \emph{quasi-synchronization almost surely} and \emph{quasi-synchronization in mean}.
\begin{defn}\label{robconsen}
Denote
$  d_{\mathcal{V}}(t)=\max\limits_{i, j\in \mathcal{V}}|x_i(t)-x_j(t)|$.
\begin{enumerate}
  \item If $\mathbb{P}\Big\{\limsup\limits_{t\rightarrow\infty}d_{\mathcal{V}}(t) \leq \epsilon\Big\}=1$, we say the system (\ref{basicHKmodel})-(\ref{Model:probnum}) achieves quasi-synchronization almost surely (a.s.);
  \item If $\limsup\limits_{t\rightarrow\infty}\textbf{E}\,d_{\mathcal{V}}(t) \leq \epsilon$, we say the system (\ref{basicHKmodel})-(\ref{Model:probnum}) achieves quasi-synchronization in mean (i.m.).\\
\end{enumerate}
\end{defn}

Different from the synchronous bounded confidence model, we can show that the asynchronous model (\ref{basicHKmodel})-(\ref{Model:probnum}) cannot achieve quasi-synchronization a.s. (Theorem \ref{Theorem:noconsas}). However, simulation studies suggest that an approximate ``synchronization'' state truly occurs for the asynchronous model. In the following section, we will reveal that this new kind of synchronization state is quasi-synchronization i.m.

\begin{rem}\label{Remark:definitionrem}
The definition of quasi-synchronization provides a normative description of an approximate synchronization in the sense of confidence threshold $\epsilon$. When a system achieves quasi-synchronization (a.s. or i.m.), all agents become neighbor to each other (a.s. or i.m.), and a cluster of social significance forms. Meanwhile, quasi-synchronization does not imply a precise error of synchronization. Actually, as it is revealed by both the previous studies about quasi-synchronization a.s. (Lemma 7 \cite{Su2017auto}, Lemma 3.2 \cite{Su2019tac}) and the present study on quasi-synchronization i.m. (Lemma \ref{Lemma:generalThm} and (\ref{Equa:bardelta})), a more accurate error of synchronization is heavily dependent on the noise amplitude $\delta$. It shows that the accurate error of synchronization can be arbitrarily small as long as $\delta$ is small enough.
\end{rem}

\section{Main Results}\label{Section:results}

In this section, we intend to introduce the main results of noise-induced synchronization of the system represented by equations (\ref{basicHKmodel})-(\ref{Model:probnum}). Firstly, we prove that the system (\ref{basicHKmodel})-(\ref{Model:probnum}) can achieve quasi-synchronization i.m. from any initial state, when there are noises with proper strength. Secondly, we testify that the system (\ref{basicHKmodel})-(\ref{Model:probnum}) cannot achieve quasi-synchronization a.s. Finally, we verify that a system achieves quasi-synchronization i.m. when it reaches quasi-synchronization a.s.

\subsection{Quasi-Synchronization i.m. of the asynchronous systems}
When there is no noise in the system (\ref{basicHKmodel})-(\ref{Model:probnum}) (i.e., $\xi_i(t)\equiv 0, i\in\mathcal{V}, t\geq 1$), the system cannot reach synchronization ($\lim_{t\rightarrow\infty}d_\mathcal{V}(t)=0$) under some initial state $x(0)\in[0,1]^n$ and $\epsilon\in(0,1)$. However, when there is noise, we can obtain the following theorems and lemmas.
\begin{thm}\label{Theorem:consmean}
Suppose that the noises $\{\xi_i(t)\}_{i\in\mathcal{V},t\geq 1}$ are zero-mean random variables with independent and identical distribution (i.i.d.), and $\textbf{E}\,\xi^2_1(1)>0, |\xi_1(1)|\leq \delta$ a.s. for $\delta>0$.
Then given any $x(0)\in [0,1]^n$ and $\epsilon\in(0,1]$, there exists $\bar{\delta}=\bar{\delta}(n,\epsilon,\alpha,p_0,p_1)>0$ such that the system (\ref{basicHKmodel})-(\ref{Model:probnum}) achieves quasi-synchronization i.m. for all $\delta\in(0,\bar{\delta}]$.
\end{thm}
By Theorem \ref{Theorem:consmean}, the quasi-synchronization i.m. of DW model can be obtained directly.
Let $\mathcal{U}(t)=\{i_1,i_2\}$ in the system (\ref{basicHKmodel})-(\ref{Model:probnum}), where $\{i_1,i_2\}\subset\mathcal{V}$ is an arbitrary choice of two agents and $\beta\in(0,1]$.  The standard DW model has a form as
\begin{equation}\label{Model:deterDW}
\begin{split}
  x_{i_r}(t+1)=&\left\{
              \begin{array}{ll}
                &\beta x_{i_r}(t)+(1-\beta)x_{i_{3-r}}(t),\,\hbox{if}\,\,|x_{i_1}(t)-x_{i_2}(t)|\leq \epsilon;\\
                &x_{i_r}(t), \,\hbox{otherwise.}
              \end{array}
            \right.\\
  x_k(t+1)=&x_k(t),\quad k\notin \mathcal{U}(t).
  \end{split}
\end{equation}
where $r=1,2$.

Let us consider the following noisy DW model
\begin{equation}\label{Model:noisyDW}
\begin{split}
  x_i(t+1)=&(\tilde{x}_i(t)+\xi_i(t+1))_{[0,1]},
  \end{split}
\end{equation}
where $\tilde{x}_i(t)$ is the right side of (\ref{Model:deterDW}).
\begin{cor}\label{Corollary:DWmean}
(Quasi-synchronization i.m. of DW model) Given any $x(0)\in[0,1]^n, \epsilon\in(0,1]$ for the system (\ref{Model:noisyDW}), there exists $\bar{\delta}=\bar{\delta}(n,\epsilon,\beta)>0$, such that $\limsup\limits_{t\rightarrow\infty}\textbf{E}\,d_{\mathcal{V}}(t) \leq \epsilon$ for all $\delta\in(0,\bar{\delta}]$.
\end{cor}

To prove Theorem \ref{Theorem:consmean}, some lemmas are needed. The following lemma provides a basic tool for analyzing the noise-induced properties of bounded confidence opinion dynamics, which roughly states that when a random walk has a uniform positive probability entering a region within a finite time,  it will almost surely enter that region in finite time.
\begin{lem}\label{Lemma:smallproboccur}
Let $\{w_t, t\geq 1\}$ be a random walk on $R^n$, $\{T_i:\Omega\rightarrow \mathbb{N}^+, i\geq 1\}$ be a sequence of increasing random variables. For $D\subset R^n$, denote $T=\inf\limits_{t\geq 1}\{t:w_t\in D\}$ and $\bar{D}=R^n-D$. If for any $T_i, i\geq 1$, there is a constant $0<p\leq 1$ such that $ \mathbb{P}\Big\{w_{T_{i+1}}\in D\Big|\bigcap_{k\leq i}\{w_{T_k}\in\bar{D}\}\Big\}\geq p$, then $\mathbb{P}\{T<\infty\}=1$.
\end{lem}
\begin{proof}
The proof of Lemma \ref{Lemma:smallproboccur} was given in the last part of Proposition 3.1 in \cite{Su2019tac}.
\end{proof}

\begin{lem}\label{Lemma:asenter}
Given any initial state $x(0)\in [0,1]^n, \epsilon\in(0,1]$ of the system (\ref{basicHKmodel})-(\ref{Model:probnum}), and any $\lambda\in(0,1]$, define $T=\inf\limits_{t\geq 0}\{t:d_{\mathcal{V}}(t)\leq \lambda\epsilon\}$, then $\mathbb{P}\{T<\infty\}=1$ for all $0<\delta\leq \frac{\lambda\epsilon}{2}$.
\end{lem}
The proof of Lemma \ref{Lemma:asenter} designs a noise protocol which drives the system to reach the goal, as in \cite{Su2017auto}. We give the proof in Appendix.

In the following, we intend to analyze the system properties of the asynchronous model (\ref{basicHKmodel})-(\ref{Model:probnum}). Specially, the analysis methodology of quasi-synchronization i.m. of the proposed asynchronous model is quite different compared to the previous studies of quasi-synchronization a.s. of synchronous models.
For $t\geq 0$, we take
\begin{equation*}
\begin{split}
 m_t\in&\Big\{i\in\mathcal{V}:x_i(t)=\min_{j\in\mathcal{V}} x_j(t)\Big\}, \\
M_t\in&\Big\{i\in\mathcal{V}:x_i(t)=\max_{j\in\mathcal{V}} x_j(t)\Big\},
\end{split}
\end{equation*}
and define the event
\begin{equation}\label{Equa:eventAt}
  A(t)=\{m_t\in\mathcal{U}(t), M_t\in\mathcal{U}(t)\}.
\end{equation}
$A(t)$ is the set of two agents with maximum and minimum opinion values who are also communicating agents at $t$. Then we have the following lemma.

\begin{lem}\label{Lemma:noiserobcons}
Consider the system (\ref{basicHKmodel})-(\ref{Model:probnum}) and denote $L_0=\frac{n(n-1)}{2}$. Given $\lambda\in(0,1]$, if there exists $T<\infty$ a.s. such that $d_\mathcal{V}(T)\leq \frac{\lambda\epsilon}{2}$, then
\begin{equation}\label{Equa:Lemmadvleqn}
  \mathbb{P}\Big\{d_\mathcal{V}(T+t)\leq d_\mathcal{V}(T)+2t\delta\leq \lambda\epsilon\Big\}=1
\end{equation}
for all $1\leq t\leq L_0, \delta\in\Big(0,\frac{\alpha\lambda\epsilon}{2n(n-1)^2}\Big]$, and
\begin{equation}\label{Equa:Lemmadvtgeqn}
  \mathbb{P}\bigg\{d_\mathcal{V}\Big(T+L_0\Big)\leq \frac{\lambda\epsilon}{2}-\frac{\alpha\lambda\epsilon}{2(n-1)}\bigg|\bigcap\limits_{r=T}^{T+L_0-1}A(r)\bigg\}=1.
\end{equation}
\end{lem}

\begin{proof}
For convenience of notation, suppose $T= 0$ a.s. To prove (\ref{Equa:Lemmadvleqn}), we only need to prove
\begin{equation}\label{Equa:dvtleqn}
 d_\mathcal{V}(t)\leq d_\mathcal{V}(0)+2t\delta\leq \lambda\epsilon, \quad a.s.
\end{equation}
for $1\leq t\leq L_0$.

Since $|\xi_i(t)|\leq \delta$ a.s., from (\ref{basicHKmodel}), $d_\mathcal{V}(t)\leq d_\mathcal{V}(t-1)+2\delta\leq\ldots\leq d_\mathcal{V}(0)+2t\delta$, implying the first part of (\ref{Equa:dvtleqn}). The second part of (\ref{Equa:dvtleqn}) can be directly obtained by $1\leq t\leq \frac{n(n-1)}{2}$, $\alpha\leq \frac{1}{n}$ and $\delta\leq \frac{\alpha\lambda\epsilon}{2n(n-1)^2}$.

Now we proceed to prove (\ref{Equa:Lemmadvtgeqn}).
By  (\ref{Model:probnum})(a) and (\ref{Equa:utik}), for $t\geq 0$
\begin{equation}\label{Equa:probAt}
  \begin{split}
  \mathbb{P}\{A(t)\}=&\sum_{k=2}^n\frac{1}{C_n^k}p_kC_{n-2}^{k-2}=\sum_{k=2}^n\frac{k(k-1)}{n(n-1)}p_k\\
  \geq &\frac{2}{n(n-1)}(p_2+\ldots+p_n)\\
=& \frac{2}{n(n-1)}(1-p_0-p_1)>0.
  \end{split}
\end{equation}
By (\ref{Model:probnum})(c), $\{A(t), t\geq 0\}$ are independent, then we can gain
\begin{equation}\label{Equa:probcapAt}
\begin{split}
  \mathbb{P}\bigg\{\bigcap\limits_{r=0}^{L_0-1}A(r)\bigg\}=\prod_{r=0}^{L_0-1}\mathbb{P}\{A(r)\}
\geq\bigg(\frac{2(1-p_0-p_1)}{n(n-1)}\bigg)^{L_0}>0.
\end{split}
\end{equation}
To prove (\ref{Equa:Lemmadvtgeqn}), we assume $\mathbb{P}\bigg\{\bigcap\limits_{r=0}^{L_0-1}A(r)\bigg\}=1$ without loss of generality, and we then need to prove
\begin{equation}\label{Equa:dvtgeqn}
 d_\mathcal{V}(L_0)\leq \frac{\lambda\epsilon}{2}-\frac{\alpha\lambda\epsilon}{2(n-1)},\quad a.s.
\end{equation}
By the above assumption, we know $\mathbb{P}\{A(0)\}=1$, which implies that $m_0$ and $M_0$ are communicating agents at $t=0$. And, $d_\mathcal{V}(0)\leq \frac{\lambda\epsilon}{2}\leq \epsilon$ implies that $m_0$ and $M_0$ are neighbors to each other. From (\ref{basicHKmodel}), we know
\begin{equation}\label{Equa:m01}
\begin{split}
x_{m_0}(1)=\alpha_{m_0}(0)x_{m_0}(0)+(1-\alpha_{m_0}(0))\frac{\sum_{j\in \mathcal{N}_{m_0}(0)}x_j(0)}{|\mathcal{N}_{m_0}(0)|}+\xi_{m_0}(1)
\end{split}
\end{equation}
then, by $\alpha\leq \alpha_{m_0}(0)\leq 1-\alpha$ in (\ref{basicHKmodelalpha}) and $|\xi_{m_0}(1)|\leq \delta$ a.s., it follows a.s.
\begin{equation}\label{Equa:m01move}
\begin{split}
  x_{m_0}(1)-x_{m_0}(0)=&(1-\alpha_{m_0}(0)\frac{\sum_{j\in \mathcal{N}_{m_0}(0)}(x_j(0)-x_{m_0}(0))}{|\mathcal{N}_{m_0}(0)|}+\xi_{m_0}(1)\\
  \geq &(1-\alpha_{m_0}(0))\frac{x_{M_0}(0)-x_{m_0}(0)}{n-1}+\xi_{m_0}(1)\\
  \geq&\frac{\alpha}{n-1} d_\mathcal{V}(0)-\delta.
  \end{split}
\end{equation}
Similarly, we can get
\begin{equation}\label{Equa:M01}
\begin{split}
x_{M_0}(1)= \alpha_{M_0}(0)x_{M_0}(0)+(1-\alpha_{M_0}(0))\frac{\sum_{j\in \mathcal{N}_{M_0}(0)}x_j(0)}{|\mathcal{N}_{M_0}(0)|}+\xi_{M_0}(1)
\end{split}
\end{equation}
and a.s.
\begin{equation}\label{Equa:M01move}
\begin{split}
  x_{M_0}(1)-x_{M_0}(0)=&(1-\alpha_{M_0}(0)\frac{\sum_{j\in \mathcal{N}_{M_0}(0)}(x_j(0)-x_{M_0}(0))}{|\mathcal{N}_{M_0}(0)|}+\xi_{M_0}(1)\\
  \leq &(1-\alpha_{M_0}(0))\frac{x_{m_0}(0)-x_{M_0}(0)}{n-1}+\xi_{M_0}(1)\\
  \leq&-\frac{\alpha}{n-1} d_\mathcal{V}(0)+\delta.
  \end{split}
\end{equation}
Equations (\ref{Equa:m01move}) and (\ref{Equa:M01move}) yield a.s.
\begin{equation}\label{Equa:M01m01dis}
\begin{split}
  |x_{M_0}(1)-x_{m_0}(1)|\leq &d_\mathcal{V}(0)-\frac{2\alpha}{n-1} d_\mathcal{V}(0)+2\delta\\
\leq &\frac{\lambda\epsilon}{2}-\Big(\frac{\alpha\lambda\epsilon}{n-1}-2\delta\Big)
\end{split}
\end{equation}
For any $i\in U(0)$, we know $x_{m_0}(0)\leq x_i(0)\leq x_{M_0}(0)$. Hence, following a similar argument as illustrated above, (\ref{Equa:M01m01dis}) implies a.s.
\begin{equation}\label{Equa:U0ijdis}
\begin{split}
  |x_i(1)-x_j(1)|\leq \frac{\lambda\epsilon}{2}-\Big(\frac{\alpha\lambda\epsilon}{n-1}-2\delta\Big)
\end{split}
\end{equation}
for any $i,j\in U(0)$.

Equation (\ref{Equa:U0ijdis}) yields that once two agents are communicating at $t=0$, their distance at $t=1$ has an upper bound which is represented by the right side of (\ref{Equa:U0ijdis}). During the following movement, their distance can exceed the upper bound only when they are not communicating agents simultaneously. In this case, their distance may increase by no more than $2\delta$ after each time step.
Since $\delta\leq \frac{\alpha\lambda\epsilon}{2n(n-1)^2}$, then a.s.
\begin{equation}\label{Equa:U0ijdist}
\begin{split}
  |x_i(t)-x_j(t)|\leq \frac{\lambda\epsilon}{2}-\Big(\frac{\alpha\lambda\epsilon}{n-1}-2t\delta\Big)\leq \frac{\lambda\epsilon}{2}-\frac{\alpha\lambda\epsilon}{2(n-1)}
\end{split}
\end{equation}
for all $i,j\in U(0)$ and $1\leq t\leq \frac{n(n-1)}{2}$.

Given any $2\leq t_0\leq L_0-1$, by (\ref{Equa:dvtleqn}), we can get $d_\mathcal{V}(t_0)\leq d_\mathcal{V}(0)+2t_0\delta$. Similar to the process of obtaining (\ref{Equa:M01m01dis}), we have a.s.
\begin{equation}\label{Equa:Mt1mt1dis}
\begin{split}
  |x_{M_{t_0}}(t_0+1)-x_{m_{t_0}}(t_0+1)|\leq &d_\mathcal{V}(t_0)-\frac{2\alpha}{n-1} d_\mathcal{V}(t_0)+2\delta\\
\leq &\frac{\lambda\epsilon}{2}-\Big(\frac{\alpha\lambda\epsilon}{n-1}-2t_0\delta\Big)+2\delta
\end{split}
\end{equation}
Then just as the process of obtaining the equation (\ref{Equa:U0ijdist}),
\begin{equation}\label{Equa:Utijdist}
\begin{split}
  |x_i(t)-x_j(t)|\leq &\frac{\lambda\epsilon}{2}-\Big(\frac{\alpha\lambda\epsilon}{n-1}-2t_0\delta\Big)+\Big(L_0-t_0\Big)(2\delta)\\
\leq& \frac{\lambda\epsilon}{2}-\frac{\alpha\lambda\epsilon}{2(n-1)}, \quad a.s.
\end{split}
\end{equation}
for all $i,j\in U(t_0)$ and $t_0+1\leq t\leq L_0$.

For any $i\in\mathcal{V}$, since $\mathbb{P}\bigg\{\bigcap\limits_{r=0}^{L_0-1}A(r)\bigg\}=1$ by assumption, it is certain that $i\in A(t)$ for some $0\leq t\leq L_0$, or $x_{m_t}(t)\leq x_i(t)\leq x_{M_t}(t)$ for all $0\leq t\leq L_0$. For both cases, by (\ref{Equa:U0ijdist}) and (\ref{Equa:Utijdist}), we have
\begin{equation}\label{Equa:Vtijdist}
\begin{split}
  \Big|x_i(L_0)-x_j(L_0)\Big|\leq \frac{\lambda\epsilon}{2}-\frac{\alpha\lambda\epsilon}{2(n-1)},\,a.s.
\end{split}
\end{equation}
for all $i,j\in\mathcal{V}$. Hence (\ref{Equa:dvtgeqn}) can be obtained.
This complete the proof.
\end{proof}
Lemma \ref{Lemma:noiserobcons} indicates that, once the system enters a region which is narrow enough, it will not get away from the region, provided some special agents with extreme opinion values are always communicating.

Furthermore, in order to achieve the final result of quasi-synchronization i.m. of the asynchronous model, we would like to introduce the following lemma. In the proof of the lemma, we introduce a stopping time as a bridge to obtain the moment property of $d_\mathcal{V}(t)$.

\begin{lem}\label{Lemma:generalThm}
Suppose the noises $\{\xi_i(t)\}_{i\in\mathcal{V},t\geq 1}$ are given in Theorem \ref{Theorem:consmean}. Let $x(0)\in [0,1]^n, \epsilon\in(0,1]$ be arbitrarily given, then for any $\mu\in(0,1]$, there is $\bar{\delta}=\bar{\delta}(\mu,n,\epsilon,\alpha,p_0, p_1)>0$, such that $\lim\limits_{t\rightarrow\infty}\textbf{E}\,d_\mathcal{V}(t)\leq \mu\epsilon$ for all $\delta\in(0,\bar{\delta}]$.
\end{lem}

\begin{proof}
Denote $\tilde{p}=\frac{2}{n(n-1)}(1-p_0-p_1)$, $L_0=\frac{n(n-1)}{2}$, $L=\min\Big\{l>0:(1-\tilde{p}^{L_0})^l\leq \frac{\mu\epsilon}{2}\Big\}$ and $T=\inf\limits_{t\geq 0}\Big\{t:d_{\mathcal{V}}(t)\leq \frac{\mu\epsilon}{4(1+L_0L)}\Big\}$, then
\begin{equation}\label{Equa:dvt}
  d_\mathcal{V}(T)\leq \frac{\mu\epsilon}{4(1+L_0L)}, a.s.
\end{equation}
Denote
\begin{equation}\label{Equa:bardelta}
\bar{\delta}=\min\Big\{\frac{\alpha\mu\epsilon}{2n(n-1)^2}, \frac{\mu\epsilon}{8(1+L_0L)}\Big\},
\end{equation}
and we next prove that
\begin{equation}\label{Equa:dvtkleqmu}
  \textbf{E}\,d_\mathcal{V}(T+k)\leq \mu\epsilon
\end{equation}
for all $k\geq 1$ and $\delta\in(0,\bar{\delta}]$.

In order to prove (\ref{Equa:dvtkleqmu}), we first consider $\mathbb{P}\Big\{d_\mathcal{V}(T+k)>\frac{\mu\epsilon}{2}\Big\}$ for any given $k>0$. Take $\lambda=\frac{\mu}{2}$ in Lemma \ref{Lemma:noiserobcons}, and notice that $d_\mathcal{V}(T)\leq \frac{\mu\epsilon}{4(1+L_0L)}\leq \frac{\mu\epsilon}{4}$ a.s., then by (\ref{Equa:Lemmadvleqn}) in Lemma \ref{Lemma:noiserobcons} and (\ref{Equa:dvt}), we can get
\begin{equation}\label{Equa:dvkleqL}
  \begin{split}
    d_\mathcal{V}(T+k)\leq  d_\mathcal{V}(T)+2k\delta  \leq  \frac{\mu\epsilon}{4(1+L_0L)}+2L_0L\bar{\delta}\leq \frac{\mu\epsilon}{4}, a.s.
  \end{split}
\end{equation}
for $1\leq k\leq L_0L$, implying
\begin{equation}\label{Equa:problessL}
\mathbb{P}\Big\{d_\mathcal{V}(T+k)>\frac{\mu\epsilon}{2}\Big\}=0, \qquad 1\leq k\leq L_0L.
\end{equation}
Next we consider $\mathbb{P}\Big\{d_\mathcal{V}(T+k)>\frac{\mu\epsilon}{2}\Big\}$ for $k>L_0L$.
Since
 $d_\mathcal{V}(T)\leq \frac{\mu\epsilon}{4(1+L_0L)}\leq \frac{\mu\epsilon}{4}$ a.s.,
by Lemma \ref{Lemma:noiserobcons}, we know that once $A(t)$ occurs  for $T\leq t\leq T+L_0-1$, then $d_\mathcal{V}(T+L_0)\leq \frac{\mu\epsilon}{4}-\frac{\alpha\mu\epsilon}{4(n-1)}$ a.s. By $\delta\leq \frac{\alpha\mu\epsilon}{8(n-1)L_0L}$, we can get
\begin{equation}\label{Equa:dvtafterAt}
  d_\mathcal{V}(T+L_0+t)\leq d_\mathcal{V}(T+L_0)+2\delta t\leq \frac{\mu\epsilon}{4}, a.s.
\end{equation}
for $1\leq t\leq L_0L$.

Denote $B(s,r)=\bigcap_{t=T+s+(r-1)L_0}^{T+s+rL_0-1}A(t), s\geq 0, r\geq 1$. (\ref{Equa:dvtafterAt}) implies that when there is a moment $T$ such that $d_\mathcal{V}(T)\leq \frac{\mu\epsilon}{4}$, and $A(t)$ occurs in the following $L_0$ times, then $d_\mathcal{V}(T+t)\leq \frac{\mu\epsilon}{4}$ for all $L_0\leq t\leq L_0+L_0L$. In other words, once $B(0,1)$ occurs, $d_\mathcal{V}(t)$ can not exceed $\frac{\mu\epsilon}{4}$ during the next $L_0L$ steps. By (\ref{Equa:dvkleqL}), $d_\mathcal{V}(T+t)\leq \frac{\mu\epsilon}{4}$ a.s. for all $1\leq t\leq L_0L$. Hence, if there is $k>L_0L$ such that $d_\mathcal{V}(T+k)>\frac{\mu\epsilon}{4}$ a.s., there must exist a period of length $L_0L$ and some integer $s\geq 0$ such that anyone of $\{B(s,1),\ldots,B(s,L)\}$ cannot happen, i.e.,
\begin{equation}\label{Equa:dvtsubsetbsL}
  \Big\{d_\mathcal{V}(T+k)>\frac{\mu\epsilon}{4}\Big\}\subset\bigg\{\bigcap_{r=1}^{L}\{\Omega -B(s,r)\}\bigg\}.
\end{equation}
By (\ref{Model:probnum}) and (\ref{Equa:probAt}), $\{A(t), t\geq 0\}$ are i.i.d., and so are $\{B(s,r), s\geq 0, r\geq 1\}$ by strong Markov property. As a result, for any given $k>L_0L$, we can get by (\ref{Equa:dvtsubsetbsL})
\begin{equation}\label{Equa:dvt1sutsetbkL}
\begin{split}
  \mathbb{P}\Big\{d_\mathcal{V}(T+k)>\frac{\mu\epsilon}{4}\Big\}\leq &\mathbb{P}\bigg\{\bigcap_{r=1}^{L}\{\Omega -B(s,r)\}\bigg\}\\
 =&(1-\mathbb{P}\{B(0,1)\})^L
\end{split}
\end{equation}
where $\Omega$ is the sample space.

By (\ref{Equa:probcapAt}) and (\ref{Equa:dvt1sutsetbkL}), we have
\begin{equation}\label{Equa:dvt1geqmu}
\begin{split}
  \mathbb{P}\Big\{d_\mathcal{V}(T+k)>\frac{\mu\epsilon}{4}\Big\}\leq&(1-\mathbb{P}\{B(0,1)\})^L\\
=&\bigg(1-\prod_{r=0}^{L_0-1}\mathbb{P}\{ A(T+r)\}\bigg)^L\\
\leq &\Big(1-\tilde{p}^{L_0}\Big)^L
\end{split}
\end{equation}
for any given $k>L_0L$.

Since $d_\mathcal{V}(t)\leq 1$ a.s. for all $t\geq 0$, by (\ref{Equa:problessL}), (\ref{Equa:dvt1geqmu}) and the definition of $L$, it follows
\begin{equation}\label{Equa:expecT}
\begin{split}
  \textbf{E}\,d_\mathcal{V}(T+k)=&\textbf{E}\Big(d_\mathcal{V}(T+k)I_{\{d_\mathcal{V}(T+k)\leq\frac{\mu\epsilon}{2}\}}+d_\mathcal{V}(T+k)I_{\{d_\mathcal{V}(T+k)>\frac{\mu\epsilon}{2}\}}\Big)\\
  \leq &\frac{\mu\epsilon}{2}+\mathbb{P}\Big\{d_\mathcal{V}(T+k)>\frac{\mu\epsilon}{2}\Big\}\\
  \leq &\frac{\mu\epsilon}{2}+(1-\tilde{p}^{L_0}\Big)^L\leq \mu\epsilon.
  \end{split}
\end{equation}
for all $k\geq 0$.

Subsequently, given any $t\geq 0$, we gain by (\ref{Equa:expecT})
\begin{equation}
\begin{split}
\textbf{E}\,d_\mathcal{V}(t)=&\textbf{E}\Big(d_\mathcal{V}(t)I_{\{T\leq t\}}+d_\mathcal{V}(t)I_{\{T>t\}}\Big)\\
=&\textbf{E}\bigg(\sum_{k=0}^{t}d_\mathcal{V}(T+k)I_{\{T=t-k\}}\bigg)+\textbf{E}\,d_\mathcal{V}(t)I_{\{T>t\}}\\
=&\sum_{k=0}^{t}\textbf{E}\,d_\mathcal{V}(T+k)I_{\{T=t-k\}}+\textbf{E}\,d_\mathcal{V}(t)I_{\{T>t\}}\\
\leq& \mu\epsilon \sum_{k=0}^{t}\textbf{E}\,I_{\{T=t-k\}}+\textbf{E}\,d_\mathcal{V}(t)I_{\{T>t\}}\\
\leq&\mu\epsilon \mathbb{P}\{T\leq t\}+\mathbb{P}\{T>t\}
\end{split}
\end{equation}
By Lemma \ref{Lemma:asenter}, $\mathbb{P}\{T<\infty\}=1$, we can thus obtain
\begin{equation}\label{limleqepsi}
\begin{split}
  \limsup\limits_{t\rightarrow\infty}\textbf{E}\,d_{\mathcal{V}}(t)\leq& \limsup\limits_{t\rightarrow\infty}\Big(\mu\epsilon \mathbb{P}\{T\leq t\}+\mathbb{P}\{T>t\}\Big)\\
  \leq &\mu\epsilon.
\end{split}
\end{equation}
This completes the proof.
\end{proof}
\begin{rem}\label{Remark:thm1rem1}
Lemma \ref{Lemma:generalThm} provides a general conclusion of noise-induced quasi-synchronization i.m. of the asynchronous system and also an estimation of the synchronization error. It shows that the synchronization error is heavily dependent on noise amplitude $\delta$. Deduced from (\ref{Equa:bardelta}), one can see that the synchronization error can be arbitrarily small (taking any $0<\mu\leq 1$) when $\delta$ is small enough.
\end{rem}
\begin{rem}\label{Remark:thm1rem2}
Lemma \ref{Lemma:generalThm} only shows that the system can achieve quasi-synchronization i.m. when noise amplitude is suitably small, while how large noise affects the synchronization is absent. Intuitively and as we can show readily, large noise could destroy the quasi-synchronization i.m. of the system. In fact, given any $0<p\leq 1$, we consider the systems with $0<\epsilon\leq 2p^2$. When $\mathbb{P}\{\xi_1(1)>\frac{\epsilon}{2p^2}\}\geq p, \mathbb{P}\{\xi_1(1)<-\frac{\epsilon}{2p^2}\}\geq p$, then for all $t\geq 1$ we can have $\mathbb{P}\{d_{\mathcal{V}}(t+1)>\frac{\epsilon}{p^2}\}\geq \mathbb{P}\{\xi_{m(t)}(t+1)<-\frac{\epsilon}{2p^2},\xi_{M(t)}(t+1)>\frac{\epsilon}{2p^2}\}\geq p^2$, implying $\textbf{E}\,d_{\mathcal{V}}(t+1)> \frac{\epsilon}{p^2}p^2=\epsilon$ for all $t\geq 1$.
\end{rem}

\noindent\emph{Proof of Theorem \ref{Theorem:consmean}:} Take $\mu=1$ in Lemma \ref{Lemma:generalThm}, and we obtain the conclusion.
\hfill $\Box$

\subsection{Quasi-synchronization a.s. of asynchronous systems}
Lemmas \ref{Lemma:asenter} and \ref{Lemma:noiserobcons} show that when the system (\ref{basicHKmodel})-(\ref{Model:probnum}) is synchronous (i.e., $|\mathcal{U}(t)|\equiv n, t\geq 1$ or $p_n$=1), the system can achieve quasi-synchronization a.s. for all $\delta\in(0,\frac{\alpha\epsilon}{n(n-1)^2})$ (this is actually the case of HK model \cite{Su2017auto}). But when the model is asynchronous, i.e., $p_n<1$, we have the following conclusion.
\begin{thm}\label{Theorem:noconsas}
Suppose that the noises $\{\xi_i(t)\}_{i\in\mathcal{V},t\geq 1}$ are zero-mean i.i.d. random variables, and $\textbf{E}\,\xi^2_1(1)>0,|\xi_i(t)|\leq \delta$ a.s. for $\delta>0$. If $p_n<1$, then the system (\ref{basicHKmodel})-(\ref{Model:probnum}) cannot achieve quasi-synchronization a.s for any $x(0)\in [0,1]^n, \epsilon\in(0,1)$  and $\delta>0$.
\end{thm}
\begin{proof}
We only need to prove that for any $x(0)\in [0,1]^n, \epsilon\in(0,1)$  and $\delta>0$, $\limsup\limits_{t\rightarrow\infty}d_\mathcal{V}(t)=1$ a.s., i.e.
\begin{equation}\label{Equa:suplimequ}
\begin{split}
  \mathbb{P}\bigg\{\bigcup\limits_{g=0}^{\infty}\{d_\mathcal{V}(t)<1, t\geq g\}\bigg\}=&1-\mathbb{P}\bigg\{\bigcap\limits_{g=0}^{\infty}\bigcup\limits_{t=g}^\infty\{d_\mathcal{V}(t)=1\}\bigg\}\\
  =&1-\mathbb{P}\Big\{\limsup\limits_{t\rightarrow\infty}d_\mathcal{V}(t)=1\Big\}=0.
  \end{split}
\end{equation}
Given any $g\geq 0$, denote $ T=\inf\limits_{t\geq g}\{t:d_{\mathcal{V}}(t)=1\}$, then by Lemma \ref{Lemma:smallproboccur}, we need to prove that, for any initial state $x(0)\in [0,1]^n$, there are $t_L>g, 0<p<1$ such that $\mathbb{P}\{d_\mathcal{V}(t_L)=1\}\geq p$.

Since $\textbf{E}\,\xi_1(1)=0, \textbf{E}\,\xi_1^2(1)>0$ and $|\xi_i(1)|\leq \delta$ a.s., there exist constants $0<a\leq \delta, 0<\bar{p}<1$ such that
\begin{equation}\label{Equa:noiseprob}
  \mathbb{P}\{a<\xi_1(1)\leq \delta\}\geq \bar{p},\,\,\mathbb{P}\{-\delta\leq\xi_1(1)<-a\}\geq \bar{p}.
\end{equation}
For $t\geq g$, consider the following noise protocol
\begin{equation}\label{Protol:noiseproto1out}
   \left\{
    \begin{array}{ll}
      &\xi_i(t+1)\in[-\delta,-a],  \quad\hbox{if}\,\, \min\limits_{j\in\mathcal{V}}x_j(t)\leq x_i(t)\leq \min\limits_{j\in\mathcal{V}}x_j(t)+\frac{d_{\mathcal{V}}(t)}{2}; \\
      &\xi_i(t+1)\in[a,\delta], \quad \hbox{if}\,\,\min\limits_{j\in\mathcal{V}}x_j(t)+\frac{d_{\mathcal{V}}(t)}{2}<x_i(t)\leq \max\limits_{j\in\mathcal{V}}x_j(t).
    \end{array}
  \right.
\end{equation}
Denote $A(t)^C=\Omega-A(t), t\geq 0$, where $A(t)$ is defined in (\ref{Equa:eventAt}), then by $p_n<1$ we have
\begin{equation}\label{Equa:probAtc}
  \begin{split}
  \mathbb{P}\{\{A(t)^C\}=&1-\mathbb{P}\{A(t)\}=1-\sum_{k=2}^n\frac{1}{C_n^k}p_kC_{n-2}^{k-2}\\
=&1-\sum_{k=2}^{n-1}\frac{k(k-1)}{n(n-1)}p_k-p_n\\
  \geq &1-p_n-\frac{n-2}{n}(p_2+\ldots+p_{n-1})\\
\geq & 1-p_n-\frac{n-2}{n}(1-p_n)\\
=&\frac{2(1-p_n)}{n}>0.
  \end{split}
\end{equation}
Hence by the independence of $\{\xi_i(t), \mathcal{U}(t), i\in\mathcal{V}, t\geq 1\}$ and (\ref{Equa:noiseprob}), we can get
\begin{equation}\label{Equa:probdvt1}
\begin{split}
  \mathbb{P}\{d_\mathcal{V}(t+1)\geq d_\mathcal{V}(t)+a\}\geq& \mathbb{P}\{A(t)^C, ~\text{protocol}~ (\ref{Protol:noiseproto1out})~ \hbox{occures at} ~t \}\\
=&\mathbb{P}\{A(t)^C\}\cdot\mathbb{P}\{\text{protocol}~ (\ref{Protol:noiseproto1out})~ \hbox{occures at} ~t\}\\
  \geq&\tilde{p}\bar{p}^n>0,
  \end{split}
\end{equation}
where $\tilde{p}=\frac{2(1-p_n)}{n}$.

Denote $t_L=\lceil\frac{1}{a}\rceil$, then under the protocol (\ref{Protol:noiseproto1out})
\begin{equation*}
  \max_{i\in\mathcal{V}}x_i(g+t_L)=1, \quad\min_{i\in\mathcal{V}}x_i(g+t_L)=0,
\end{equation*}
yielding $d_\mathcal{V}(g+t_L)=1$.
By (\ref{Equa:probdvt1}) and the independence of $\{\xi_i(t), \mathcal{U}(t), i\in\mathcal{V}, t\geq 0\}$, we gain
\begin{equation}\label{Equa:probdVLgeq1}
\begin{split}
 \mathbb{P}\{d_\mathcal{V}(g+t_L)=1\} \geq &\prod_{t=g+1}^{g+t_L}\mathbb{P}\{A(t)^C, ~\text{ protocol (\ref{Protol:noiseproto1out}) occurs at}~t\}\\
 \geq& \tilde{p}^{t_L}\bar{p}^{nt_L}>0.
 \end{split}
\end{equation}
Let $p=\tilde{p}^{t_L}\bar{p}^{nt_L}$, and this completes the proof.
\end{proof}

Theorem \ref{Theorem:noconsas} shows that quite different from the synchronous model, the asynchronous model cannot achieve quasi-synchronization a.s., no matter how small the non-zero noise is.

\subsection{Quasi-synchronization a.s. implies quasi-synchronization i.m.}
Theorems \ref{Theorem:consmean} and \ref{Theorem:noconsas} indicate that, for the system (\ref{basicHKmodel})-(\ref{Model:probnum}), quasi-synchronization i.m. does not necessarily imply quasi-synchronization a.s. The following theorem reveals that the converse conclusion is true, i.e., quasi-synchronization a.s. leads to quasi-synchronization i.m.
\begin{thm}\label{Theorem:asimplymean}
If $\mathbb{P}\Big\{\limsup\limits_{t\rightarrow\infty}d_{\mathcal{V}}(t) \leq \epsilon\Big\}=1$, then $\limsup\limits_{t\rightarrow\infty}\textbf{E}\,d_{\mathcal{V}}(t) \leq \epsilon$.
\end{thm}
\emph{Proof:}
Since $\mathbb{P}\Big\{\limsup\limits_{t\rightarrow\infty}d_{\mathcal{V}}(t) > \epsilon\Big\}=0$,
\begin{equation*}
  \mathbb{P}\bigg\{\lim\limits_{s\rightarrow\infty}\bigcup_{t\geq s}\{d_{\mathcal{V}}(t) > \epsilon\}\bigg\}=\lim\limits_{s\rightarrow\infty}\mathbb{P}\bigg\{\bigcup_{t\geq s}\{d_{\mathcal{V}}(t) > \epsilon\}\bigg\}=0,
\end{equation*}
where the first equation follows from the exchange theorem of limit operation and probability measure (refer to Corollary 1.5.2 \cite{Chowprob}).
By $d_\mathcal{V}(t)\leq 1$ a.s., we obtain
\begin{equation}
\begin{split}
\textbf{E}\,d_\mathcal{V}(t)&=\textbf{E}\Big(d_\mathcal{V}(t)I_{\{d_\mathcal{V}(t)\leq \epsilon\}}+d_\mathcal{V}(t)I_{\{d_\mathcal{V}(t)>\epsilon\}}\Big)\\
&\leq\epsilon \mathbb{P}\{d_\mathcal{V}(t)\leq \epsilon\}+\mathbb{P}\{d_\mathcal{V}(t)>\epsilon\}.
\end{split}
\end{equation}
Consequently, we can get
\begin{equation*}
  \begin{split}
     \limsup\limits_{t\rightarrow\infty}\textbf{E}\,d_{\mathcal{V}}(t)\leq&\limsup\limits_{t\rightarrow\infty}\Big(\epsilon \mathbb{P}\{d_\mathcal{V}(t)\leq \epsilon\}+\mathbb{P}\{d_\mathcal{V}(t)>\epsilon\}\Big)\\
     \leq& \epsilon.
   \end{split}
\end{equation*}
This completes the proof.
\hfill $\Box$

\section{Simulations}\label{Section:simulations}

In this section, we introduce our simulation result to help understand the meaning of quasi-synchronization i.m.
Let $n= 40, \epsilon=0.1, \alpha_i(t)=\frac{1}{|\mathcal{N}_i(t)|+1}$, and $|\mathcal{U}(t)|$ is randomly selected from $0,1,\ldots,n$ with equal probability at each time. The initial opinion values are randomly generated on $[0,1]$, then we add independent noises with uniform distribution on $[-\delta, \delta]$ to the system (\ref{basicHKmodel})-(\ref{Model:probnum}). Take $\delta=0.01$, then Fig. \ref{Fig:meanconsen} shows that the system achieves a synchronization at about $t=10000$. However, it is not the almost sure quasi-synchronization (refer to \cite{Su2017auto} for the simulation study of almost sure quasi-synchronization), since it can be calculated that $\max_t d_\mathcal{V}(t)=0.1105>\epsilon=0.1$ from $t=15000$ to $40000$. Such an approximate synchronization can be measured by quasi-synchronization i.m.

\begin{figure}[htp]
  \centering
  \includegraphics[width=3.5in]{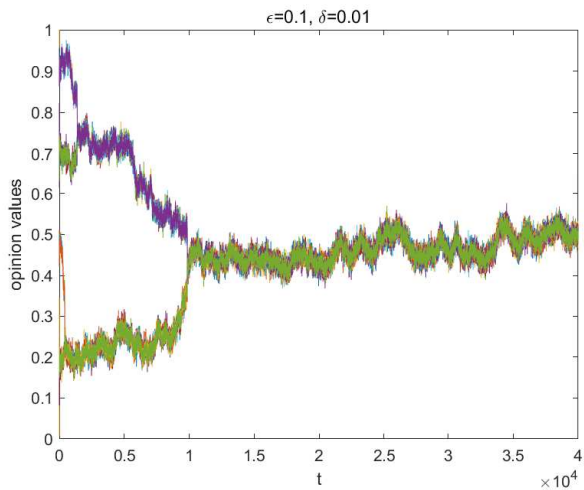}\\
  \caption{Opinion evolution of system (\ref{basicHKmodel})-(\ref{Model:probnum}) of 40 agents. The initial system states are randomly generated on  $[0,1]$, confidence threshold $\epsilon=0.1$, noise strength $\delta=0.01$. }\label{Fig:meanconsen}
\end{figure}

\section{Conclusions and discussions}\label{Section:conclusions}

In this paper, we studied the noise-induced synchronization of an asynchronous opinion dynamics of bounded confidence. A noisy stochastic asynchronous model was proposed, and we proved that, though the proposed model can not achieve quasi-synchronization a.s. as the synchronous HK model does, it can achieve quasi-synchronization i.m. The results for the first time prove theoretically the noise-driven synchronization of the DW model, which has been observed by previous simulation studies.
Moreover, quasi-synchronization i.m. was verified to be weaker than quasi-synchronization a.s. The results in this paper help complete the theory of noise-induced synchronization of bounded confidence dynamics. Moreover, due to the limitations of applying traditional control method, which relies heavily on the accurate information of systems states, to the control of complex systems, the present results lay a theoretical foundation for developing noise-based control strategy of more general complex social opinion systems.

Recently, Li et al designed a particle robotics system which used stochastic movement of loosely coupled robotic systems to realize a global coordination \cite{Li2019nature}. For a further interpretation of this study, we want to mention that the results in this paper also highlight the noise-driven properties of loosely coupled self-organizing particle robotic systems.

\section*{Appendix}

\emph{Proof of Lemma \ref{Lemma:smallproboccur}:}
Notice (\ref{Equa:noiseprob}) and consider the following noise protocol: for all $i\in\mathcal{V}$, $t\geq 0$
\begin{equation}\label{Protol:noiseproto1in}
  \left\{
    \begin{array}{ll}
     & \xi_i(t+1)\in[a,\delta],  \quad\hbox{if}\,\,\min\limits_{j\in\mathcal{V}}x_j(t)\leq\widetilde{x}_i(t)\leq \min\limits_{j\in\mathcal{V}}x_j(t)+\frac{d_{\mathcal{V}}(t)}{2}; \\
      &\xi_i(t+1)\in[-\delta,-a], \quad\hbox{if}\,\, \min\limits_{j\in\mathcal{V}}x_j(t)+\frac{d_{\mathcal{V}}(t)}{2}<\widetilde{x}_i(t)\leq \max\limits_{j\in\mathcal{V}}x_j(t),
    \end{array}
  \right.
\end{equation}
where
\begin{equation}\label{tildex}
  \widetilde{x}_i(t)=\left\{
    \begin{array}{ll}
      &\alpha_i(t)x_i(t)+(1-\alpha_i(t))\frac{\sum_{j\in\mathcal{N}_i(t)}x_j(t)}{|\mathcal{N}_i(t)|},  \quad\hbox{if}\,\, i\in\mathcal{U}(t)~\hbox{and}~\mathcal{N}_i(t)\neq \emptyset; \\
      &x_i(t),  \qquad \hbox{otherwise}.
    \end{array}
  \right.
\end{equation}
Since $\delta\leq \frac{\lambda\epsilon}{2}$, following a similar argument of the proof of Theorem 5 in \cite{Su2017auto}, there exists $L> 0$ such that
\begin{equation}\label{Equa:probevent}
\begin{split}
  \mathbb{P}\bigg\{d_\mathcal{V}(mL)\leq \lambda\epsilon\Big|\bigcap\limits_{j<m}d_\mathcal{V}(jL)> \lambda\epsilon\bigg\}\geq \bar{p}^L
  \end{split}
\end{equation}
for $m\geq 1$. Let $D=\{x\in R^n: \max_{i,j}|x_i-x_j|\leq \lambda\epsilon\}, T_k=kL, p=\bar{p}^L$, then by (\ref{Equa:probevent}) and Lemma \ref{Lemma:smallproboccur}, $\mathbb{P}\{T<\infty\}=1$.
\hfill $\Box$

\end{document}